\documentclass[12pt]{ifacconf}

\usepackage{amssymb,amsmath,amsfonts,eurosym,geometry,ulem,graphicx,caption,color,setspace,comment,caption,natbib,pdflscape,subfigure,array}
\usepackage{xcolor}
\usepackage[utf8]{inputenc}
\usepackage{url}
\usepackage{multirow}
\usepackage{mathtools}
\newcommand{\mR}{\mathbb{R}}

\newcommand{\mc}{\mathcal}

\newcommand{\Tha}{\Theta}

\renewcommand{\leq}{\leqslant}
\renewcommand{\geq}{\geqslant}
\mathtoolsset{showonlyrefs}

\usepackage{algpseudocode}
\usepackage{algorithmicx}
\usepackage{algorithm}
\usepackage{caption}
\usepackage{eufrak}

\newtheorem{theo}{Theorem}
\newtheorem{Def}{Definition}
\newtheorem{assume}{Assumption}

\newtheorem{lemma}{Lemma}

\DeclareMathOperator{\eps}{\epsilon}

\newcommand{\ket}{|\varphi\rangle}
\newcommand{\bra}{\langle\varphi|}

\newcommand{\ketn}{|\varphi^{(n)}\rangle}
\newcommand{\bran}{\langle\varphi^{(n)}|}

\newcommand{\ketpsi}{|\psi\rangle}

\newcommand{\keti}{|\varphi_i\rangle}
\newcommand{\brai}{\langle\varphi_i|}
\newcommand{\ketk}{|k\rangle}
\newcommand{\brak}{\langle k|}

\newcommand{\ketjp}{|\psi_{j'}\rangle}
\newcommand{\brajp}{\langle\psi_{j'}|}

\newcommand{\tr}{\text{Tr}}

\newenvironment{proof}[1][Proof]{\noindent\textbf{#1.} }{\ \rule{0.5em}{0.5em}}

\newcolumntype{L}[1]{>{\raggedright\let\newline\\arraybackslash\hspace{0pt}}m{#1}}
\newcolumntype{C}[1]{>{\centering\let\newline\\arraybackslash\hspace{0pt}}m{#1}}
\newcolumntype{R}[1]{>{\raggedleft\let\newline\\arraybackslash\hspace{0pt}}m{#1}}

\begin{document}
\begin{frontmatter}

\title{Finitely Repeated Adversarial Quantum Hypothesis Testing} 


\author[First]{Yinan Hu} 
\author[First]{Quanyan Zhu} 

\address[First]{Department of Electrical and Computer Engineering, New York University, Brooklyn, NY, 11201, USA (e-mail: \{yh1817, qz494\}@nyu.edu)}

\begin{abstract}                
We formulate a passive quantum detector based on a quantum hypothesis testing framework under the setting of finite sample size. In particular, we exploit the fundamental limits of performance of the passive quantum detector asymptotically. Under the assumption that the attacker adopts separable optimal strategies, we derive that the worst-case average error bound converges to zero exponentially in terms of the number of repeated observations, which serves as a variation of quantum Sanov's theorem.  We illustrate the general decaying results of miss rate numerically, depicting that the `naive' detector manages to achieve a miss rate and a false alarm rate both exponentially decaying to zero given infinitely many quantum states, although the miss rate decays to zero at a much slower rate than a quantum non-adversarial counterpart. Finally we adopt our formulations upon a case study of detection with quantum radars.    
\end{abstract}

\end{frontmatter}

\section{Introduction}
The study of quantum systems has caught increasing attention both in theories and applications since the construction and manipulation of entangled quantum states based on quantum circuits of a variety of realizations such as photons \cite{wang2015entangled_photon_angular_momentum,wang2016experimental_ten_photon}, ion traps \cite{pogorelov2021compact_ion_traps}  and nuclear magnetic resonances (NMR) \cite{mkadzik2022precision_nmr}, have reached significant breakthroughs during the past several decades. The detection and discrimination of quantum states play a crucial role in building quantum systems. 

One possible way to countermeasure the attack on the detection process is to collect more samples from the target quantum state before arriving at a decision. Such methods lead to the formulations of sequential quantum hypothesis testing \cite{vargas2021quantum_seq_hypo_testing} and of fixed sample size quantum hypothesis testing. In fixed sample size hypothesis testing, the detector is given finite, independent, and repeated observations from the same  before making a decision.  The sequential quantum hypothesis testing framework \cite{vargas2021quantum_seq_hypo_testing} inherits Wald's sequential analysis \cite{wald2004sequential}, where instead of accepting all of the samples at the same time, the detectors can stop accepting samples further when the predetermined error rates are satisfied.  

Although sequential hypothesis testing requires in general fewer samples than fixed sample size tests to reach the same error rate \cite{levy2008principles_signal_detection} and could be implemented online, the framework of hypothesis testing with finite repeated observations still admits its own advantages, especially for analytical purposes. There are several large-deviation theorems such as quantum versions of Stein's lemma \cite{hiai1991quantum_stein_lemma} and Sanov's theorem \cite{li2014quantum_asymptotic_hypo_sanov} characterizing that a quantum detector can achieve the exponential decay of both false alarm rate and miss rate to zero when the sample size becomes infinitely large.   

Challenges also arise when uncertainty exists in  quantum systems. One type of uncertainties is quantum noise in the environment resulting in the undesired correlation of the target system with the environment and incurring a loss of information of the target state due to decoherence \cite{schlosshauer2007decoherence}. Another type of challenges is   adversarial attacks against the target quantum system, the transmission channel, or the measurement scheme. For instance, authors in \cite{Portmann_2022_security_quantum_cryptography} summarize recent attacks and their corresponding countermeasures in quantum cryptography systems. Authors in \cite{blakely2022quantum_limits_spoofing} have studied the quantum limitations of detection upon classical spoofing attacks when the samples are electromagnetic signals. Thus when designing a quantum system that detects and discriminates quantum states, it is crucial to take into consideration malicious attacks and design countermeasures to enhance the detection performances under tainted quantum states. 

To this end, we develop a game-theoretic framework to study how an attacker (He) could affect the asymptotic performance of a quantum detector (She) under repeated observations. Game-theoretic formulations have been consolidated into classical hypothesis testing formulations to study the security of detection schemes under strategically contaminated samples \cite{hu2022game_NP}\cite{yasodharan2019hypo_test_game}. In a generic binary quantum hypothesis scenario, the detector knows that the target systems are characterized by an unknown density matrix that fits the ones underlying the hypothesis $H_0$ or $H_1$. The detector is given $n$ non-coherent quantum states drawn from the same unknown density state before making a decision on the true hypothesis. With adversaries, those quantum states are simultaneously manipulated by the attacker before being delivered to the passive quantum detector. We assume that the detector is naive about the distortion and applies a quantum hypothesis testing framework to decide the true hypothesis $H_0$ or $H_1$ and thus formulate the relationship between the passive detector and the attacker as a Stackelberg game \cite{Stackelberg2010market}.       

Our theoretical and numerical analyses show that the detector  still manages to reach an exponential decay to zero for both types of errors. However, the rate at which the miss rate converges to zero under adversaries is much smaller than the one under the adversary-free scenario. If we introduce local robustness condition, assuming the attacker's quantum distortion can be constrained by a norm, then we manage to develop worst-case convergence rate of the miss rate characterized by distributionally robust optimization \cite{rahimian2019distributionally_review} upon quantum relative entropy \cite{von2018mathematical_QM}.     

The rest of the paper is organized as follows. In section \ref{sec:format} we briefly review the formulation of the passive quantum detector that we developed in \cite{hu2022quantum_MITHM}. Later we extend the formulation into quantum passive detectors with repeated observations. In section \ref{sec:error_bound} we analyze the asymptotic performances of the passive quantum detector. In section \ref{sec:radar_case_study} we apply our adversarial sequential quantum hypothesis testing formulation in a case study of quantum radar detectors, which send out repeated pulses to detect the absence or presence of a target object under spoofing attacks.   
We conclude in section \ref{sec:conclusion}.
\subsection{Notations} 
We apply the traditional Dirac's bra-ket notations  \cite{dirac1981principles} to denote generic quantum states: $\bra\in\mc{H}^*,\ket \in\mc{H}$, where $\mc{H}$ represents a Hilbert space (over the reals $\mR$) with dimension $d$. For arbitrary $\ket, \ketpsi\in\mc{H}$, we denote their inner product as $\langle \psi\ket$. The norm of $\varphi$ is defined as $\|\varphi\| = \sqrt{\langle \varphi\ket}$.  
 Let $B(\mc{H})$ be the space of all positive, Hermitian, and bounded operators from $\mc{H}$ to itself, with their norms defined as their largest absolute value of the eigenvalues $\|\cdot\|_{1}$. Let $\mc{S}$ be the subset of $B(\mc{H})$ such trace of its operators is $1$. The space of all positive-operator valued measurements is denoted as $\mc{V}$.

Furthermore, we introduce the notation $(\cdot)^{\otimes n}$ to characterize the $n$-fold tensor product of the set or the quantum state. For example, $B(\mc{H})^{\otimes n}$ means applying tensor product of $B(\mc{H})^{\otimes n}$ upon itself $n$ times. We define $\mc{S}^{\otimes n}, \mc{V}^{\otimes n}$ in a similar way. 

\section{The passive quantum detector with fixed sample size}
\label{sec:format}
In this section, we formulate the passive quantum detector given multiple copies from the same unknown target system.  Let $\rho\in \mc{S}$ be the density operator of the unknown target system. We assume the scenario of binary hypothesis testing \cite{helstrom1969quantum}, that is, only one density matrix is included under $H_0,H_1$ respectively:
\begin{equation}
    H_0:\rho= \rho_0,\; H_1:\rho = \rho_1,
    \label{hypo_rho1_rho0}
\end{equation}
where 
\begin{equation}
\rho_0= \sum_{j = 1}^d{r^0_j\ketjp\brajp},\;\;\rho_1 = \sum_{i = 1}^d{r^1_i\keti\brai}, 
\end{equation}
with $\sum_{i}{r^1_i} = \sum_{j}{r^1_j} = 1,\;r^1_j,r^0_i\geq 0$ and $\ketjp,\keti\in \mc{H},\; i,j=1,2,\dots,d$. The two sets of quantum states  $\{\ketjp\},\{\keti\}$ span the same Hilbert space $\mc{H}$ of dimension $d$, but we assume that $\{\ketjp\}^d_{j=1},\{\keti\}^d_{i=1}$ are two distinct bases, hence the two density operators $\rho_0,\rho_1$ do not necessarily commute. 
 
Under repeated observations, the target system produces $n$ quantum states $\ketn\in\mc{H}^{\otimes n}$ which are collected by a detector (He), who wants to arrive at a decision rule $\delta\in \mc{H}^{\otimes n}\rightarrow [0,1]$ on the genuine characterization of the target system in \eqref{hypo_rho1_rho0}: $\delta(\varphi^{(n)}) = k,\;k\in\{0,1\},$ when he thinks the hypothesis $H_k$ holds true. According to the measurement postulate of quantum mechanics \cite{von2018mathematical_QM}, the detector decides by applying a measurement $\Pi^n_1\in\mc{V}^{\otimes n}$ upon the received quantum state $\ketn$:
\begin{equation}
    \delta(\varphi^{(n)}) = \bran \Pi^n_1 \ketn.
\end{equation}

The decision rule $\delta$, or equivalently the measurement $\Pi^n_1$, leads to a detection rate $P^n_D: \mc{V}\rightarrow [0,1]$  and a false alarm rate $P^n_F:\mc{V}\rightarrow [0,1]$ as follows:
\begin{align}
     P^n_D(\Pi^n_1) = \tr(\Pi^n_1\rho^{\otimes n}_1), \;\; P^n_F(\Pi^n_1) = \tr(\Pi^n_1\rho^{\otimes n}_0).
     \label{PD_PF}
\end{align}

A passive quantum detector does not know that the quantum states have been distorted. Hence we assume that she does not change her decision rule based on the quantum states she receives. 
Thus we let the detector's cost function $u_B:\mc{V}\times \mc{S}\times \mc{S}\rightarrow \mR$ be the Bayesian risk, which is a weighting sum of \textit{counterfactual} miss rate and \textit{counterfactual} false alarm rate in \eqref{PD_PF} as if the quantum state were untainted:    
\begin{equation}
\begin{aligned}
    \underset{\Pi^n_1\in\mc{V}^{\otimes n}}{\min}&u^{(n)}_B(\Pi^n_1,\rho'^{\otimes n}_1,\rho'^{\otimes n}_0) \\
    \Leftrightarrow  \underset{\Pi^n_1\in\mc{V}^{\otimes n}}{\min}& c_{1,n} \tr((\mathbf{1}-\Pi^n_1)\rho^{
    \otimes n}_1) + c_{0,n}\tr(\Pi^{n}_1\rho^{\otimes n}_0),\;\;
\label{detector_stackelberg_passive}
\end{aligned}
\end{equation}
where $c_{0,n},c_{1,n}$ only depend on the number of observations $n$. Referring to \cite{helstrom1969quantum} and denoting $\tau(n) = \frac{c_{1,n}}{c_{0,n}}$, the detector's optimal solution measurement $\Pi^{n*}_1$ can be designed as follows:
\begin{equation}
    \Pi^{n*}_1(\tau) = \sum_{\gamma_j>0}{|\gamma_j\rangle\langle \gamma_j|},
    \label{sol_stackelberg_detector}
\end{equation}
where $\{|\gamma_j\}^{d^n}_{j=1}$ are orthogonal eigenstates of $\rho^{\otimes n}_1-\tau(n)\rho^{\otimes n}_0$ with eigenvalues $\gamma_j$.

The detector's decision threshold $\tau(n)$ only depends on the number of quantum states received. We could set the threshold $\tau: [N]\rightarrow \mR_+$
\begin{equation}
    \tau(n) = \frac{\bar{c}_0 n + \bar{d}_0}{\bar{c}_1 n + \bar{d}_1},\;n\in[N],
\end{equation}
where $\bar{c}_0,\bar{c}_1,\bar{d}_0,\bar{d}_1>0$ are constants. The quantum detector develops a passive decision rule $\bar{\delta}^{(n)}:\mc{H}^{\otimes n}\rightarrow [0,1]$ through a collective measurement upon the $n$ quantum states:
\begin{equation}
    \bar{\delta}^{(n)}(\ketn) = \bran\Pi^{n*}_1\ketn,
    \label{detector_decision_rule}
\end{equation}
We now formulate the attacker's actions. Similar to the space of quantum operations on a single quantum state discussed in \cite{hu2022quantum_MITHM}, we formulate the attacker's actions as the space of collective quantum noisy operations $(E^{0,n},E^{1,n})$ upon the $n$-fold quantum system observed from the user ($\rho_0$ or $\rho_1$) as follows. 
\begin{equation}
\begin{aligned}
 \mc{E}^{(n)} &= \{(E^{0,n}, E^{1,n})\in {B}(\mc{H})^{\otimes 2n}:  \\
 &\sum_{k}{E^{l,n}_kE^{l,n\dagger}_k} = \mathbf{1},\;l\in\{0,1\}\}.
 \label{action_space_attacker_Alice}
\end{aligned}
\end{equation}

\begin{figure}
    \centering
    \includegraphics[scale=0.27]{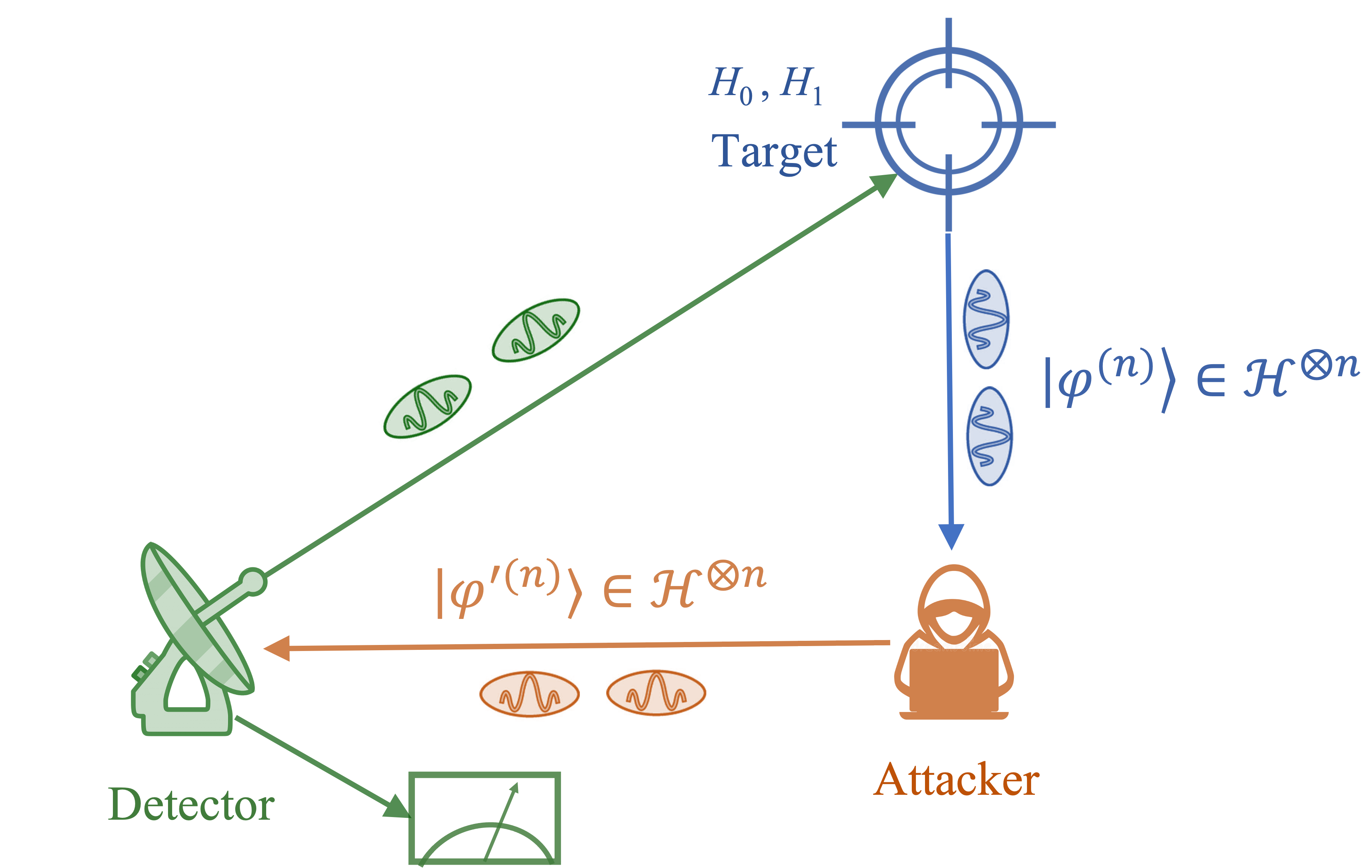}
    \caption{The scheme of adversarial quantum state detection with repeated observations. The target system produces $n$ quantum states $\ketn$ from different target systems depending on true hypothesis $H_0$ or $H_1$.  The attacker intervenes and substitutes the $n$ clean quantum states with the manipulated states. One case study is the target detection of quantum radars (e.g., \cite{slepyan2021quantum_radar_lidar}) under spoofing attacks.}
    \label{fig:quantum_man_in_the_middle}
\end{figure}

Knowing detector's strategy \eqref{sol_stackelberg_detector}, the detector designs her optimal strategies $\rho^{(n)*}_1,\rho^{(n)*}_0$ by minimizing her utility function $u^{(n)}_A: \mc{S}^{\otimes n}\times \mc{S}^{\otimes n}\times \mc{V}^{\otimes n}\rightarrow \mR$ as follows:
\begin{equation}
\begin{aligned}
&\underset{\substack{\rho^{(n)}_0 \in \mc{S}^{\otimes n} \\\rho^{(n)}_1\in\mc{S}^{\otimes n}}}{\min}  u^{(n)}_A(\Pi^{n*}_1,\rho^{ (n)}_1,\rho^{(n)}_0)\\
\Longleftrightarrow & \underset{\substack{\rho^{(n)}_0 \in \mc{S}^{\otimes n} \\\rho^{(n)}_1\in\mc{S}^{\otimes n}}}{\min}\tr(\Pi^{n*}_1\rho^{(n)}_1) + \beta_n S(\rho^{(n)}_1\|\rho^{\otimes n}_1)  \\
&\qquad\quad+\beta_n  S(\rho^{( n)}_0\|\rho^{\otimes n}_0),
    \label{attacker_stackelberg}
\end{aligned}
\end{equation}
 where we adopt the Von-Neumann relative entropy \cite{nielsen2002quantum_information} for any two density operators $\nu_1,\nu_0\in\mc{S}^{\otimes n}$ for any $n\in [N]$ as
\begin{equation}
    S(\nu_1\|\nu_0) := \tr(\nu_1(\ln \nu_1-\ln\nu_0))
\end{equation}
and $\Pi^{n*}_1$ is the solution to detector's optimization problem \eqref{detector_stackelberg_passive}. The regularization parameter $\beta_n$ depends on $n$ only. For simplicity, we set $\beta_n = \lambda^n$ with $\lambda\geq 0$ being a constant. 

Referring to \eqref{detector_stackelberg_passive} and \eqref{attacker_stackelberg}, we develop the Stackelberg game $\mc{G}^n$ between the attacker and the passive quantum detector as follows.
\begin{Def}[Game of passive detector]
We define the relationship between the passive quantum detector and the attacker as a Stackelberg game $\mc{G}^n$ (e.g., \cite{Stackelberg2010market}) with the following tuples
\begin{equation}
    \mc{G}^n = \langle \mc{I}, \Tha,  \mc{E}^{(n)}, \mc{V}^{\otimes n}, u^{(n)}_A, u^{(n)}_B \rangle,
\end{equation}
where $\mc{I}= \{\text{Attacker},\;\text{Detector}\}$ refers to the set of players; $\Tha = \{H_0,H_1\}$ refers to the set of true hypotheses specified in \eqref{hypo_rho1_rho0}; The set $\mc{E}^{(n)}$ in \eqref{action_space_attacker_Alice} characterizes the attacker's strategy space; the set of measurements $\mc{V}^{\otimes n}$ characterizes the detector's space of decision rules. The functions $u^{(n)}_A, u^{(n)}_B$ specified in \eqref{detector_stackelberg_passive} and \eqref{attacker_stackelberg} are the objectives of the attacker and the detector, respectively. 
\label{def:stackelberg_detector}
\end{Def}

\subsection{Equilibrium analysis}

Notice in the game $\mc{G}^n$ the attacker and the detector adopt sequential rationality to design their optimal strategies \cite{fudenberg1998game}. The formulations \eqref{detector_decision_rule} and \eqref{sol_stackelberg_detector} characeterize the detector's optimal decision rules. For the attacker's optimal strategies, i.e. optimal distorted density matrices $\rho^{(n)*}_0,\rho^{(n)*}_1\in S^{\otimes n}$, we state the following theorem.
\begin{prop}[Attacker's optimal strategies]
\label{attacker_strategies_stackelberg_prop}
Let $\mc{G}$ be the Stackelberg game between the attacker and the detector (the passive quantum detector) as mentioned in definition \ref{def:stackelberg_detector}. Then we can derive the attacker's optimal strategies, characterized by distorted density operators $\rho'^{(n) *}_1, \rho'^{(n) *}_0 $  as follows:
\begin{align}
    \rho^{(n)*}_0&=\rho^{\otimes n}_0,
    \label{sol0_attacker_Stackelberg}  \\
    \rho^{(n)*}_1&=  \frac{\exp(\ln\rho^{\otimes n}_1 - \frac{1}{\beta_n}\Pi^{n*}_1) }{\tr(\exp(\ln\rho^{\otimes n}_1 - \frac{1}{\beta_n}\Pi^{n*}_1))}.
    \label{sol1_attacker_Stackelberg}    
\end{align}

\end{prop}

\section{Error bounds under separable attacks}
\label{sec:error_bound}

In the previous section we study the attacker's optimal strategies upon distorting $n$ copies of the target system. In this section we study the quantum detector's detection performance on the $n$-fold product of the target system under adversaries. We first recall some fundamental theorems on large deviations of quantum hypothesis testing under adversary-free scenarios such as quantum Sanov's theorem \cite{li2014quantum_asymptotic_hypo_sanov} and quantum Stein's lemma \cite{hiai1991quantum_stein_lemma}. In this section, we aim to develop a worst-case bound under the transformation of measures of the quantum states under each hypothesis. 

Notice that in general the attacker's distortions characterized in \eqref{sol0_attacker_Stackelberg}\eqref{sol1_attacker_Stackelberg} can not be decomposed into tensor products of elements in $\mc{S}$. For the sake of analyzing the adversarial detection performance, we introduce extra assumptions to obtain separable optimal strategies: 
\begin{lemma}
Let $\rho^{(n)*}_0,\rho^{(n)*}_1\in\mc{S}^{\otimes n}$ be the attacker's optimal density operators in \eqref{sol0_attacker_Stackelberg} and \eqref{sol1_attacker_Stackelberg}. Then if $\rho^{\otimes n}_1 - \frac{1}{\beta_n}\Pi^{n*}_1$ has only separable quantum states as its eigenvectors, then the attacker's optimal strategies \eqref{sol0_attacker_Stackelberg}\eqref{sol1_attacker_Stackelberg} are separable, that is, there exist $\alpha_{(j)}\in \mc{S},\;j\in[n]$ such that 
\begin{equation}
\begin{aligned}
    \rho^{(n)*}_1& =  \frac{\exp(\ln\rho^{\otimes n}_1 - \frac{1}{\beta_{n}}\Pi^{ n*}_1) }{\tr(\exp(\ln\rho^{\otimes n}_1 - \frac{1}{\beta_{n}}\Pi^{ n*}_1))}\\ & = \alpha_{(1)}\otimes \dots \alpha_{(n)}.
    \label{separability}
\end{aligned}
\end{equation}
\end{lemma}

\begin{proof}
We could write $\rho^{\otimes n}_1 - \frac{1}{\beta_n}\Pi^{n*}_1$ in the form of eigen-decomposition given all of its eigenvectors are separable:
\begin{equation}
\begin{aligned}
  &\rho^{\otimes n}_1 - \frac{1}{\beta_n}\Pi^{n*}_1\\ 
  &=\sum_{j_1,\dots, j_n}{p_{j_1,\dots, j_n}|\varphi_{j_1}\rangle\dots|\varphi_{j_n}\rangle\langle\varphi_{j_1}|\dots \langle\varphi_{j_n}|}.
  \label{eigenvectors_separable}
\end{aligned}
\end{equation}
We can then construct $\alpha_{(1)}\dots \alpha_{(n)}$ by taking partial traces upon the right hand side of \eqref{eigenvectors_separable}.
\end{proof}

We now analyze the impact of the attacker's strategy, which is characterized by $\rho^{(n)*}_0,\rho^{(n)*}_1$ in \eqref{sol0_attacker_Stackelberg}\eqref{sol1_attacker_Stackelberg} upon the detector's performance under $n$ repeated observations. We assume lemma 1 holds and for simplicity, further assume also that the attacker applies the same distortion (denoted as $E_1,E_0\in\mc{E}^{(1)}$) for every of the $n$ quantum states.   Since the attacker's strategies could depend on the parameter $\lambda$, the distorted density operator $\rho'^*_1$ could still lie anywhere. As a result, we could consider the detector's worst-case performance asymptotically.  We first introduce the extension of detection rate $P^n_D:\mc{V}^{\otimes n}\rightarrow [0,1]$ and false alarm rate $P^n_F:\mc{V}^{\otimes n}\rightarrow [0,1]$ under measurements $\Pi^n_1$ upon $n$ repeated observations as 
\begin{equation}
    P^n_D(\Pi^n_1) = \tr(\Pi^n_1\rho^{(n)}_1),\;\;P^n_F(\Pi^n_1) = \tr(\Pi^n_1\rho^{(n)}_0).
    \label{detection_rate_n}
\end{equation}

The quantum Wasserstein distance of order 1 is generalized from the classical Wasserstein distance \cite{villani2008optimal}.
We adopt the metric to describe the Wasserstein distance of order 1 \cite{de2021quantum_wasserstein_distance_1} between two quantum density matrices $\sigma_1,\sigma_0$ as follows:

\begin{Def}[Quantum Wasserstein distance]
Let $\sigma_1,\sigma_0\in\mc{S}^{\otimes n}$. Then we can define the quantum Wasserstein distance between $\sigma_1,\sigma_0$ as 
\begin{equation}
\begin{aligned}
   &{W}_1(\sigma_1,\sigma_0) := \min \Big(\sum_{i=1}^{n}{c_i}:c_i\geq 0,\\
   &\;\sigma_1-\sigma_0 = \sum_{i=1}^n{\sigma^{(i)}_1 - \sigma^{(i)}_0},\sigma^{(i)}_1,\sigma^{(i)}_0\in S^{\otimes n},\;\\
   &\;\tr_i\sigma^{(i)}_1 =\tr_i\sigma^{(i)}_0 \Big)   
\end{aligned}
\end{equation}
\end{Def}

The separability of the attacker's optimal strategies cannot guarantee an upper bound of the asymptotic detection rate. From \eqref{sol1_attacker_Stackelberg} we know when the regularization parameter $\beta_n \rightarrow 0$ the attacker could distort the quantum states so that the genuine detection rate is arbitrarily small. To obtain a non-vacuous error bound we constrain the attacker's quantum operations by a norm:
\begin{assume}[`local robustness']
\label{local_robust}
Let $(E_0,E_1)$ be a pair of generic attacker's distortion operations under $n=1$ in \eqref{action_space_attacker_Alice}. Then we call them locally robust with tolerance $\eps$ if the following holds:
\begin{equation}
\begin{aligned}
(E_1,E_0)\in\mc{W},\;\mc{W}=\{ \|{E}_j-\mathbf{1}\|_{\diamond}\leq \eps,\;j=0,1\}.
    \label{eq:local_robustness}
    \end{aligned}
\end{equation}
where $\|\cdot\|_{\diamond}$ refers to the diamond norm for quantum operations \cite{watrous2018theory_quantum_computation_diamond_norm}.
\end{assume}

\begin{theo}
 Assume the attacker's distorted quantum system $\rho^{(n)*}$ meets \eqref{separability} with $\alpha_{(1)}\dots \alpha_{(n)}$ be the optimal strategies for every copy of the true quantum state. Let $\{(E^i_0,E^i_1)\}^n_{i=1}$ be the corresponding sequence of the attacker's operations. Assume also that the attacker's quantum operations meet assumption \ref{local_robust}. Let $P^n_D,P^n_F$ be the detection rates and false alarm rates defined in \eqref{detection_rate_n} and define
 \begin{equation}
     \mc{R} = \{\hat{\rho}^{(n)}\in \mc{S}^{\otimes n}:\; \hat{\rho}^{(n)}_1-\tau(n)\rho^{\otimes n}_0 = \Pi^{n*}_1\}.
 \end{equation}
 
 Then the worst-case error bound can be characterized as follows:
\begin{equation}
\begin{aligned}
   \underset{\substack{(E^i_0,E^i_1)\in\mc{W} \\ i=1,2,\dots,n}}{\sup}\;\underset{n\rightarrow \infty}{\lim}\;&\frac{1}{n}\log(1 - P^n_D(\Pi^{n*}_1)) \\ 
   &\leq -\underset{\substack{W_1(\rho'_1,\rho_1)\leq \eps \\ \rho^{(n)}_1\in \mc{R}}}{\inf}S(\rho^{(n)}_1\|\rho'_1)).
\end{aligned}
\end{equation}
\end{theo}

\begin{proof}
We first introduce the following lemma:
\begin{lemma}
\label{lemma: wasserstein_bound_quantum_operator}
If $\|{E}_1 - \mathbf{1}\|_{\diamond}\leq \eps$, then ${W}_1(\mc{E}(\rho_1),\rho_1)\leq \eps$. 
\end{lemma}

{\textbf{Proof of Lemma \ref{lemma: wasserstein_bound_quantum_operator}}:}
We refer to the definition of Wasserstein distance ${W}_{1}$  between the two measures $\rho'_1$ and $\rho$ as  
\begin{equation}
    W_{1}({E}_1(\rho_1),\rho_1) = \inf\{\|\alpha_1-\rho_1\|_1,\;{E}_1(\rho_1) = \alpha_1\}. 
    \label{def:wasserstein_infty}
\end{equation}
Considering $\|{E}_1(\rho_1)-\rho_1\|_1\leq \eps$ and \eqref{def:wasserstein_infty} yields the conclusion.

We now prove the theorem. By quantum stein's lemma \cite{hiai1991quantum_stein_lemma} and quantum Sanov's theorem \cite{li2014second_quantum_sanov} we have
\begin{equation}
\begin{aligned}
   \underset{n\rightarrow\infty}{\lim}\;  \frac{1}{n}\log (1- P^n_D(\Pi^n_1))=-\underset{\substack{\rho^{(n)}_1\in \mc{R}}}{\inf}S(\rho^{(n)}_1\|\rho'_1),
\end{aligned}
\end{equation}
We can replace the density operator $\rho_1$ above with any other distorted operator $\rho'_1\in\mc{S}$ such that $W_1(\rho'_1,\rho_1)\leq \eps$. We pick $\alpha_{(j^*)}$, the one inducing the largest Wasserstein distance Taking supremum on both sides over all density operators $\rho'^*_1$ that characterize the attacker's optimal strategies and that satisfy the local robustness condition  \eqref{eq:local_robustness}, we have 
\begin{equation}
\begin{aligned}
   &\underset{W_1(\alpha_{(j)*},\rho_1)\leq\eps}{\sup}\underset{n\rightarrow\infty}{\lim}\;  \frac{1}{n}\log \tr((\mathbf{1} -\Pi^n_1) \rho'^{*\otimes n}_1) \\ 
   &\leq \underset{\substack{W_1(\rho'^*_1,\rho_1)\leq \eps \\ \rho^{(n)}_1\in \mc{R}}}{\sup} -S(\rho^{(n)}_1\|\rho'_1) \\
   &=-\underset{\substack{W_1(\rho'^*_1,\rho_1)\leq \eps \\ \rho^{(n)}_1\in \mc{R}}}{\inf} S(\rho^{(n)}_1\|\rho'_1)
\end{aligned}
\end{equation}
as stated in the theorem. 
\end{proof}

\begin{figure}
    \centering
    \includegraphics[scale=0.3]{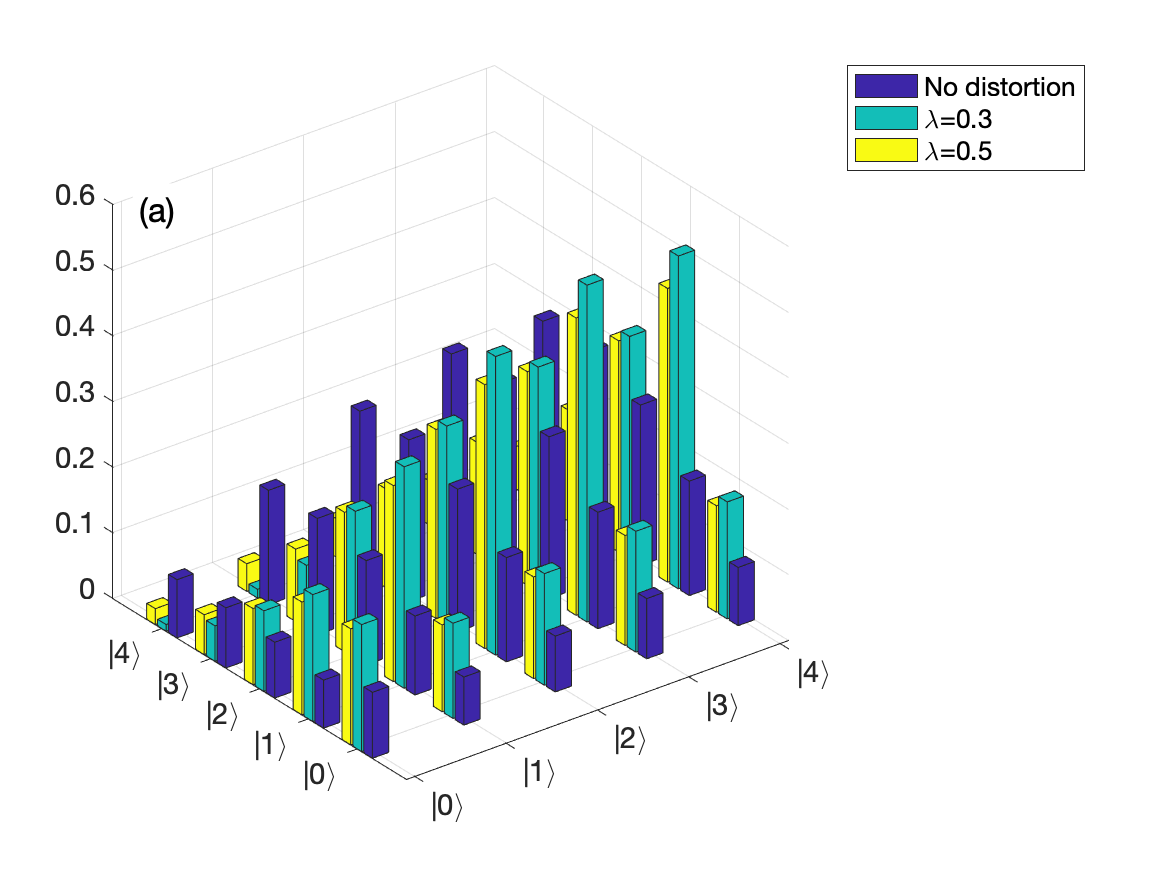} \\
    \includegraphics[scale=0.3]{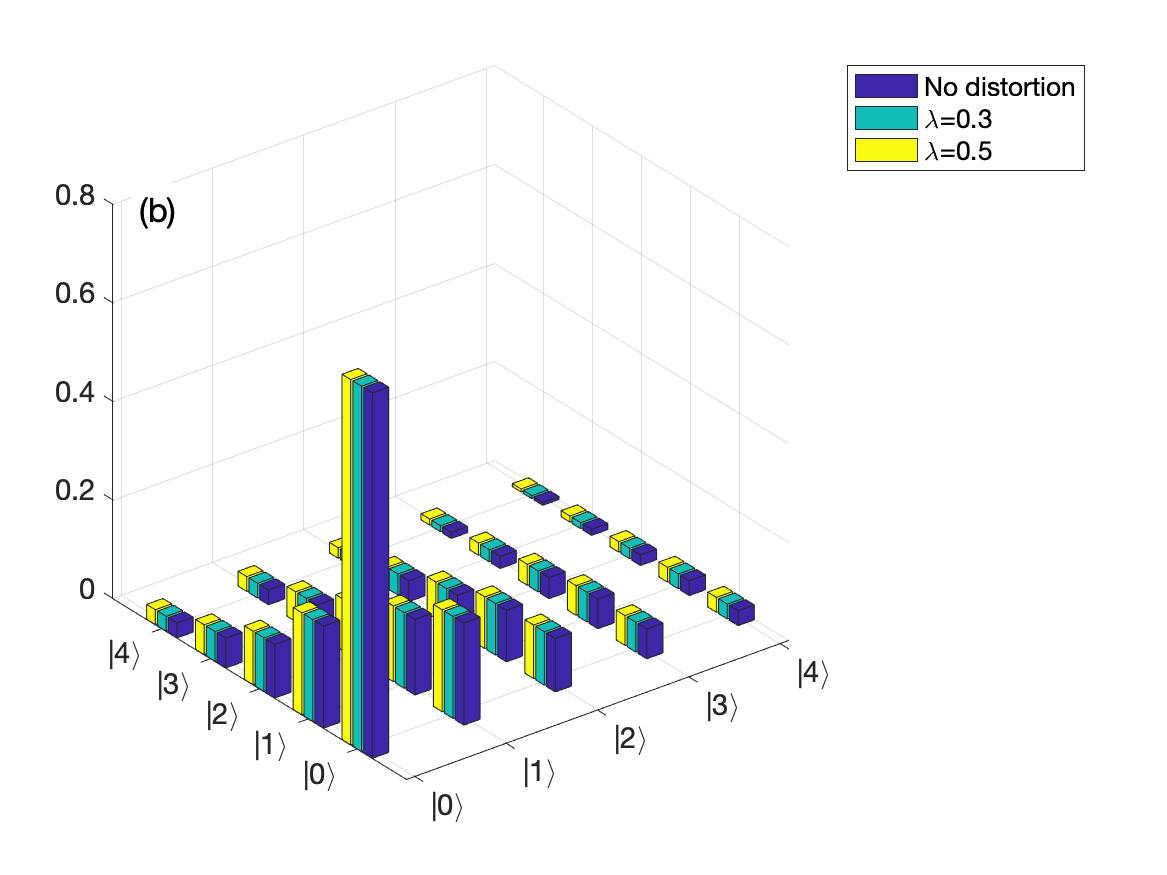} 
    \caption{Attacker's optimal strategies $\rho'^*_1$ (upper) and $\rho'^*_0$ (lower) under different choices of $\lambda$. The genuine quantum states under the hypothesis $H_0,H_1$ are defined as Gaussian states in \eqref{H_0_H1_target_no_target}. }
    \label{fig:attacker_strategy_bar3}
\end{figure}

\section{Case Study: spoofing in quantum radars}
\label{sec:radar_case_study}
In this section, we apply our formulation of passive quantum detectors under repeated observations upon the scenario of detection with quantum radars. As depicted in Figure  \ref{fig:quantum_man_in_the_middle}, a quantum radar \cite{torrome2020quantum_radar} sends out photonic quantum states into the atmosphere and collects reflective signals with its sensor. If the target object is present ($H_1$), the reflective signals will consist of photonic states representing coherent states with positive photon numbers. Otherwise, the reflective signals will resemble a vacuum state, implying the object is absent $(H_0)$. The detection process is affected by the existence of the environmental noise states, usually resulting mainly from the background light, that contaminates the reflective signals under both hypotheses.

There are currently four types of conceptual quantum radars, depending on whether they exploit entanglement and correlation of photonic quantum states. We assume that the quantum radar produces coherent, non-entangled quantum photonic states \cite{torrome2020quantum_radar}. For simplicity, we apply finite-dimensional truncated Gaussian states \cite{miranowicz1994coherent_finite_dimension} with $K$ as the maximum number of photons: 
 \begin{equation}
        |\zeta\rangle = \frac{e^{-|\zeta|^2}}{Z_K}\sum_{i=0}^{K}{\frac{\zeta^i}{\sqrt{i!}}|i)},
        \label{number_photon_states}
    \end{equation}
where $|i)\in\mc{H}$ stands for the $i$ number states of photons with the normalization factor $Z_K$. We adopt Lloyd's quantum illumination protocol \cite{lloyd2008quantum_illumination} to formulate the quantum noise state as a coherent state represented by a large mean photon number. As a consequence, we can $\rho_0,\rho_1$ as the density matrices corresponding to the hypotheses $H_0,H_1$ as follows: 
\begin{equation}
    \label{H_0_H1_target_no_target}
   \begin{aligned}
  H_0:\rho_0 &= (1-N_B)|0\rangle\langle 0| + N_B{\ketk\brak}, 
   \\
  H_1:\rho_1 &= (1-x)\left((1-N_B)|0\rangle\langle 0| + N_B{\ketk\brak}\right) \\
    &\quad+ x|l\rangle\langle l|,
\end{aligned} 
\end{equation}
where $x\in[0,1]$ refers to the reflective index, $N_B\in[0,1]$ characterizes the noise level of the environment, and $k$ characterizes the mean photon number of the noise state described in \eqref{number_photon_states}.

An attacker aims to launch a spoofing attack to undermine the performance of the detector  by intercepting the transmission of reflective quantum signals and sending manipulated signals to the quantum radar detector as depicted in Figure \ref{fig:quantum_man_in_the_middle}. We assume the quantum radar makes decisions `naively' as in \eqref{detector_stackelberg_passive} and the attacker correspondingly distorts the $n$ quantum states according to \eqref{sol1_attacker_Stackelberg}.   
In \cite{hu2022quantum_MITHM} we have illustrated the performance of a passive quantum detector in terms of receiver-operational characteristic curve (ROC curve) \cite{levy2008principles_signal_detection}. In this paper, we aim to find detection rates $P^n_F$ and false alarm rates $P^n_F$ in terms of the number of quantum states received under adversaries.  

We choose $N_B = 0.4, l =2, k = 1, x= 0.8$ in \eqref{H_0_H1_target_no_target} and $\tau(n) = \frac{0.7n + 1.5}{n+1} $, where $n$ denotes the number of pulses of reflective signals received by the radar. In Figure  \ref{fig:attacker_strategy_bar3} we plot the attacker's optimal strategies in terms of density matrices together with the original density states under $n=1$ with different choices of $\lambda$. We can observe that the attacker does not distort $\rho_0$, as she aims at minimizing detection rates and does not focus on false alarm rates. On the other hand, the attacker attenuates the quantum states representing lower mean photon numbers and enhances those with higher photon numbers in $\rho_1$ since quantum states with lower mean photon numbers are more likely leading to a decision of $\rho_1$. We observe a similar phenomenon when illustrating the attacker's optimal strategies when given multiple non-coherent quantum states.    

In Figure \ref{fig:seq_hypo_passive_detector}, we plot the detection rates and false alarm rates in terms of the number of copies of quantum states $N$ that the detector accepts when the quantum states undergo adversarial attacks of different magnitudes that are parameterized by $\lambda$ (or equivalently, $\beta_n$). The scenario $\lambda=0$ refers to a distortion-free case which can be used as a baseline.
We saw in the plot that both the false alarm rate and the miss rate (equal to $1-P^n_F(\Pi^{n*}_1)$) converge to zero  exponentially yet at different rates. Lower choices of $\lambda$ (equivalently, $\beta_n$) imply less cost for the attacker in distorting states characterized in \eqref{attacker_stackelberg}, leading to a slower exponential rate of decaying of miss rates, which is consistent with our analysis developed in section \ref{sec:error_bound}. In the meantime, we observe that the false alarm rates converge to zero exponentially at the same rate, in spite of the attacker's regularization parameter $\lambda$ since the attacker does not distort $\rho_0$. 

\begin{figure}
    \centering
    \includegraphics[scale=0.25]{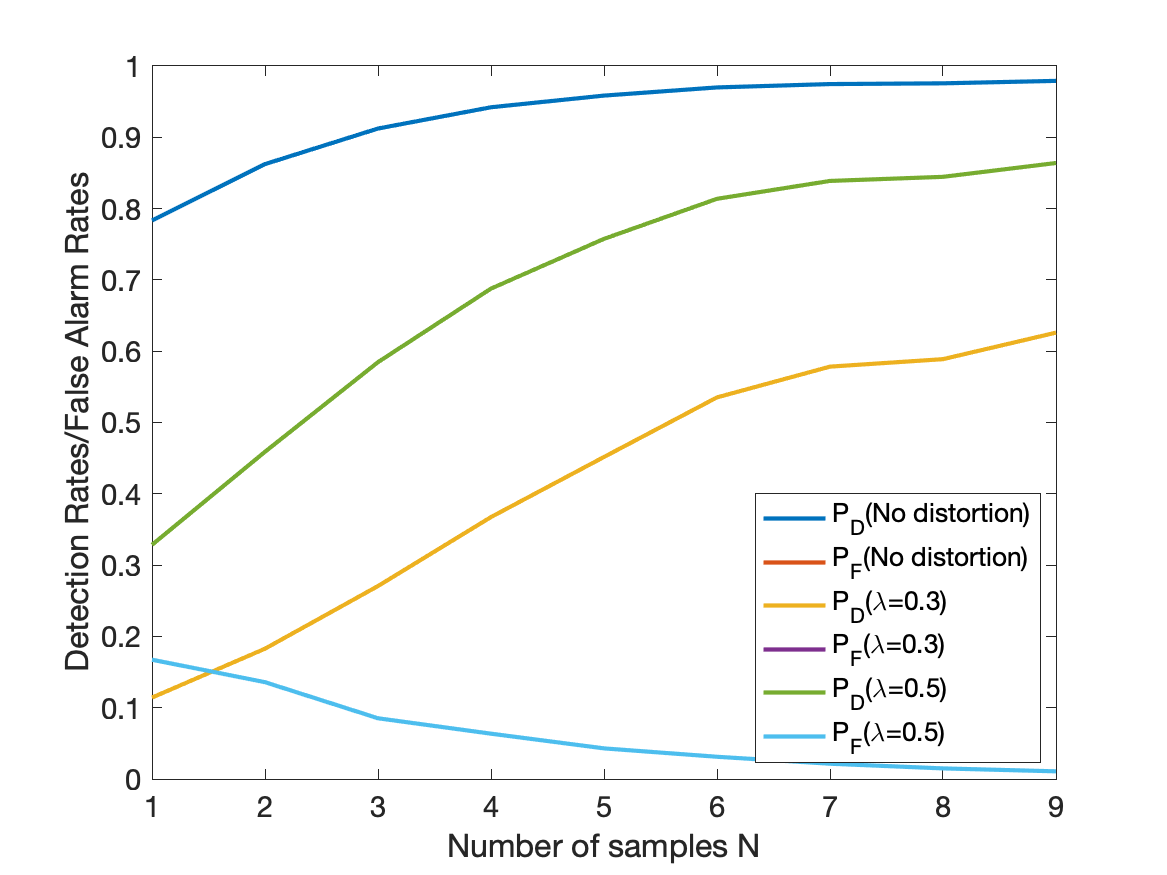}
    \caption{The detection rates $P^n_D$ and false alarm rates $P^n_F$ of a passive quantum detector in terms of the number of repeated observations. The attacker manipulates the quantum states based on \eqref{sol1_attacker_Stackelberg}.   }
    \label{fig:seq_hypo_passive_detector}
\end{figure}

\section{Conclusion}
\label{sec:conclusion}
We have formulated a passive quantum detector based on a quantum hypothesis testing framework under repeated observations. We have also studied the fundamental limits of performance of the passive quantum detector asymptotically and characterized the results as a worst-case asymptotic error bound, which serves as a variation of quantum Sanov's theorem.  Our numerical results illustrate that the `naive' detector manages to achieve a miss rate and a false alarm rate both exponentially decaying to zero given infinitely many quantum states, although the miss rate decays to zero at a much slower rate than a quantum non-adversarial counterpart.  

One possible direction to extend our work includes sequential adversarial hypothesis testing, i.e., considering the number of samples needed to achieve the desired detection rate or the false alarm rate, in particular under the scenario where observations incur costs themselves. Also, the attacker may be constrained to launch attacks obeying some simple pattern, thus the detector could reach a lower bound under those constraints for the attacker. In addition, we have shown a conceptual application of our formulation in studying adversarial attacks on target detection systems using quantum radars. It is promising to investigate a hands-on implementation of our framework to enhance the security of radar systems if we specify the radar systems with further details.    

\bibliographystyle{IEEEtran}
\bibliography{rl}
\appendix 
\section{The proof of proposition 1}
\begin{proof}
We obtain the attacker's optimal strategies by solving the optimization problem \eqref{attacker_stackelberg}. We apply first-order conditions by taking partial derivatives of $u^{(n)}_A$ in terms of $\rho^{(n)}_0,\rho^{(n)}_1$ and set them to be zero, respectively:
\begin{align}
    0 &\equiv \frac{\partial u^{(n)}_A}{\partial \rho^{(n)}_1} = \Pi^{n*}_1 + \beta_n (\ln(\rho^{(n)}_1) - \ln(\rho^{\otimes n}_1)) \\
      0 &\equiv \frac{\partial u^{(n)}_A}{\partial \rho^{(n)}_0} =  \beta_n (\ln(\rho^{(n)}_0) - \ln(\rho^{\otimes n}_0)). 
\end{align}
Thus without considering the constraints $\rho^{(n)}_0,\rho^{(n)}_1\in \mc{S}^{\otimes n}$, we arrive $\rho^{(n)*}_1 \sim {\exp(\ln\rho^{\otimes n}_1 - \frac{1}{\beta_n}\Pi^{n*}_1)},\;\rho^{(n)*}_0 \sim \rho^{\otimes n}_0$, taking into account the trace constraints of density operators $\rho^{(n)*}_1,\rho^{(n)*}_0$ yields the solution in \eqref{sol0_attacker_Stackelberg}\eqref{sol1_attacker_Stackelberg}.
\end{proof}
\end{document}